\documentclass[smallextended]{svjour3}
\pdfoutput=1
\usepackage{hyperref}
\usepackage[square, numbers, comma]{natbib}
\usepackage{geometry}
\usepackage{amsmath, amssymb}
\usepackage{color}
\usepackage{graphicx}
\usepackage{pgf}
\usepackage{tikz}
\usetikzlibrary{shapes.geometric, arrows, positioning}
\usepackage{enumerate}
\usepackage{caption}
\usepackage{booktabs}
\usepackage[T1]{fontenc}
\usepackage{algorithm}% http://ctan.org/pkg/algorithms
\usepackage{algpseudocode}% http://ctan.org/pkg/algorithmicx

\newcommand{\program}{P}
\newcommand{\virtualhardware}{\mathcal{T}}
\newcommand{\hardware}{H}

\newcommand{\problemfont}{\textsc}
\newcommand{\algorithmfont}{\textsf}
\newcommand{\chimera}{$\mathcal{C}_{L,M,N}$}
\newcommand{\chimeraargs}[3]{$\mathcal{C}_{#1, #2, #3}$}
\usepackage{mathtools}

\DeclareMathOperator*{\argmin}{arg\,min}
\newtheorem{subroutine}{Subroutine}

\algrenewtext{Function}[2]{\algorithmicfunction\ \algorithmfont{#1}(#2)}
\newcommand{\Callsf}[2]{\algorithmfont{#1}(#2)}
\newcommand{\return}[1]{\textbf{return} #1}

\usepackage{enumitem}

\usepackage{soul}

\title{Optimizing Adiabatic Quantum Program Compilation using a Graph-Theoretic Framework}
\author{Timothy~D.~Goodrich \and Blair~D.~Sullivan \and Travis~S.~Humble}

\institute{
T.D. Goodrich \and B.D. Sullivan \at
Department of Computer Science,
North Carolina State University,
Raleigh, North Carolina 27606, USA.\\
\email{Corresponding author: tdgoodri@ncsu.edu}
\and
T.S. Humble \at
Quantum Computing Institute,
Oak  Ridge  National  Laboratory,
1  Bethel  Valley  Road,
Oak  Ridge,  Tennessee  37831,  USA.\\
}

\begin{document}

\maketitle

\begin{abstract}
Adiabatic quantum computing has evolved in recent years from a theoretical
field into an immensely practical area, a change partially sparked by D-Wave
System's quantum annealing hardware. These multimillion-dollar quantum
annealers offer the potential to solve optimization problems millions of times
faster than classical heuristics, prompting researchers at Google, NASA and
Lockheed Martin to study how these computers can be applied to complex
real-world problems such as NASA rover missions. Unfortunately, compiling
(embedding) an optimization problem into the annealing hardware is itself a
difficult optimization problem and a major bottleneck currently preventing
widespread adoption. Additionally, while finding a single embedding is
difficult, no generalized method is known for tuning embeddings to use minimal
hardware resources. To address these barriers, we introduce a graph-theoretic
framework for developing structured embedding algorithms. Using this framework,
we introduce a biclique virtual hardware layer to provide a simplified
interface to the physical hardware. Additionally, we exploit bipartite
structure in quantum programs using odd cycle transversal (OCT) decompositions.
By coupling an OCT-based embedding algorithm with new, generalized reduction
methods, we develop a new baseline for embedding a wide range of optimization
problems into fault-free D-Wave annealing hardware. To encourage the reuse and
extension of these techniques, we provide an implementation of the framework
and embedding algorithms.
\end{abstract}

\section{Introduction}
Adiabatic quantum computing (AQC) is a model of computation that utilizes
quantum mechanics to solve difficult optimization problems.
As originally proposed by Farhi et al.~\cite{farhi2000quantum}, AQC relies on
the dynamical evolution of a quantum state under a Hamiltonian that changes
adiabatically from an initial to final form.
This computational model uses the final Hamiltonian to express an optimization
problem such that adiabatic evolution will recover the corresponding ground
state.

In its most general form, the AQC model is equivalent to other universal
quantum computing models. However, any limitation on the Hamiltonian forms may
reduce the power of the computational model.
Recently, an embodiment of the AQC model with a restricted Hamiltonian was
developed using superconducting flux qubits by D-Wave Systems Inc.
This quantum processor provides a large number of addressable qubits (up to
2048 in the latest D-Wave 2000Q processor) that implement a programmable Ising
model over a restricted geometry.
While not a universal quantum computer, the D-Wave processor has been
shown to produce quantum effects and yield time-to-solution orders of magnitude
faster than classical algorithms \cite{denchev2016what, king2017quantum}.
Use of this \emph{quantum annealer} \cite{kadowaki1998quantum} has evolved
beyond the design stage to testing and
deployment, with recent applications including computational chemistry, NP-hard
graph problems, image recognition, and more \cite{denchev2016what,
kassal2011simulating, lucas2014ising, neven2008image, rieffel2015case,
venturelli2015quantum}.

A key step in using current AQC-based processors is compiling the executable
program that will run on hardware with restricted connectivity
\cite{humble2014integrated, britt2015high}.
Both the problem and hardware layouts are conventionally represented using
graphs with the problem defined by variables connected with dependencies and
the hardware layouts defined by qubits connected with couplers.
Under this graph-theoretic formulation, the compilation process reduces to the
NP-hard problem of \emph{minor embedding} the problem graph into the hardware
graph. In practice, this step represents a limitation bottleneck for the
end-to-end program performance because existing embedding algorithms take
orders of magnitude longer to execute than the quantum annealer itself
\cite{humble2016performance}. Furthermore, no efficient universal embedding
algorithm exists, with past algorithms addressing specific classes of problem
instance (e.g. complete graphs, very sparse graphs, etc.) and hardware instance
(e.g. D-Wave Chimera graph, etc.), along with a myriad of additional
assumptions (e.g. fault-free hardware, parameter values, etc.). However, given
the disjoint development of algorithms for these specialized instances, the
resulting techniques cannot be combined in a common framework.

To address this incompatibility, we introduce a graph-theoretic framework for
developing tuned and modularized embedding algorithms. This framework
introduces the concept of a \emph{virtual hardware} graph that provides a
judiciously
simplified representation of the physical hardware graph, greatly reducing the
complexity of embedding subroutines. Many existing embedding algorithms are
compatible with the virtual hardware layer and we rewrite them as modular
subroutines. We then introduce generalized reduction subroutines for minimizing
the hardware footprint of a given embedding. We are able to apply these
reduction subroutines to the embedding algorithms emulated by our framework,
producing notable improvements for reducing hardware footprint.

As a proof of concept, we provide a complete bipartite virtual hardware
compatible with the D-Wave Chimera hardware structure. By exploiting bipartite
problem structure with an odd cycle transversal decomposition (OCT), we are
able to provide embeddings for edge-dense problem graphs. We additionally
present a linear-time approximation algorithm for computing OCT decompositions,
leading to fast embedding algorithms. Further use cases are provided by
Hamilton and Humble \cite{hamilton2016identifying}.

Finally, we provide an efficient implementation of the full virtual hardware
framework, including new and existing embedding and reduction subroutines,
available at \url{https://github.com/TheoryInPractice/aqc-virtual-embedding}.
Experimentally, we find that this framework is able
to unify and expand on existing embedding algorithms, providing baseline tools
for future development. Further, we find that OCT-based embedding algorithms
perform better -- in run time, size of problem graph embedded, and number of
qubits used -- than the existing TRIAD and CMR algorithms
\cite{cai2014practical, choi2011minor}.

The manuscript is organized as follows: Section 2 introduces adiabatic quantum
computing and the D-Wave hardware, including an overview of related work.
Section 3 defines virtual hardware and a stack of baseline subroutines -- the
graph-theoretic framework -- and details the emulation and enhancement of
existing embedding algorithms using our framework. Section 4 introduces a new
embedding subroutine that exploits bipartite structure in problem and hardware
graphs by using an odd cycle transversal decomposition and biclique virtual
hardware, respectively; we additionally present a new, fast approximation
algorithm for computing an odd cycle transversal. Section 5 contains
experimental results of embedding algorithms detailed in previous sections.
Finally, we summarize, present our conclusions, and outline future work in
Section 6.

\section{Background}

\begin{figure}[t]
\centering
\includegraphics[width=0.65\textwidth]{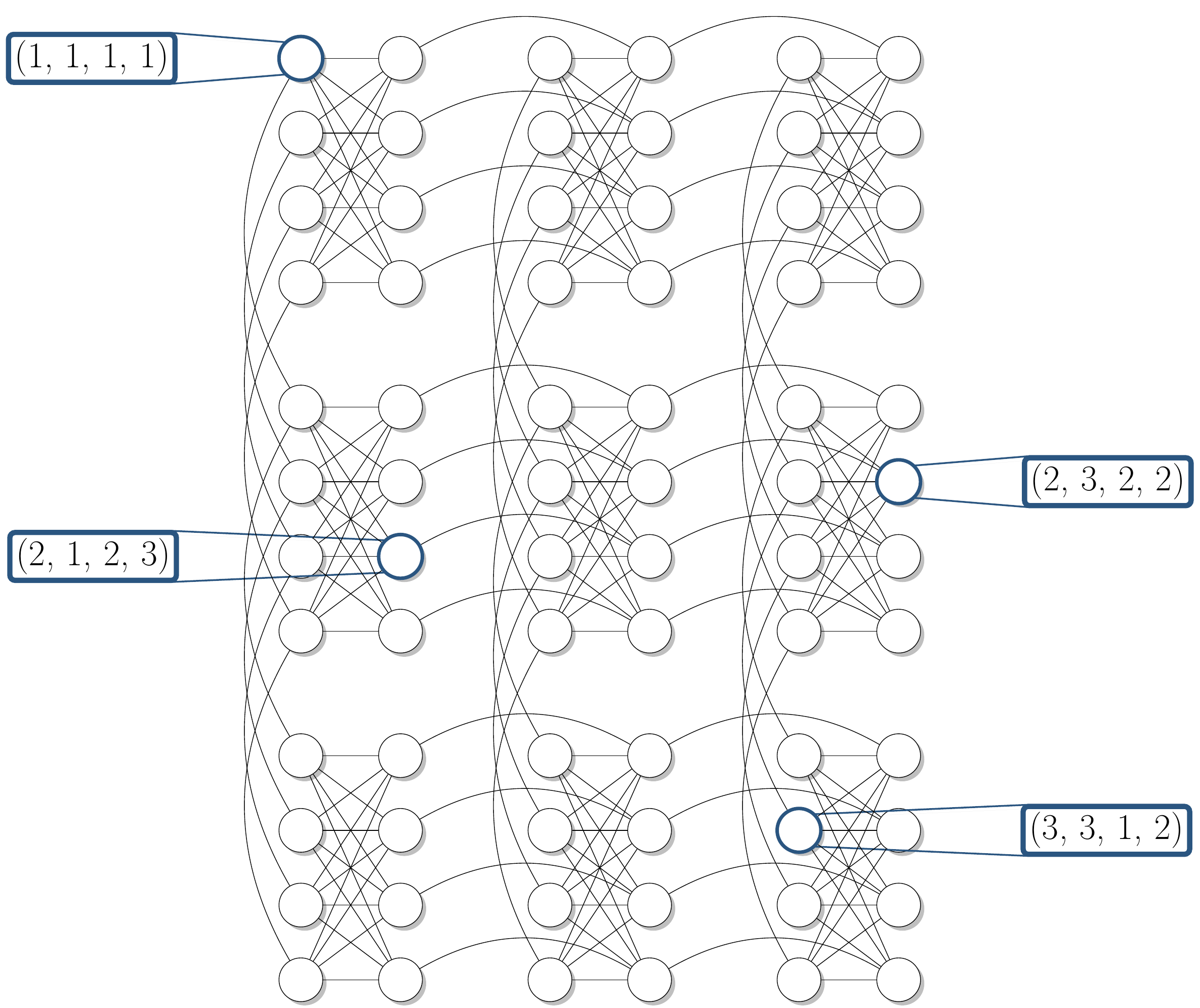}
\caption{A Chimera \chimeraargs{3}{3}{4} graph. The Chimera location labels of
four qubits are highlighted.}
\label{figure:chimera}
\end{figure}

We assume graphs are simple and undirected. For a graph $G$, we denote its
vertices with $V(G)$ and edges with $E(G)$, and let $n = |V(G)|$ and
$m = |E(G)|$ if the graph in question is clear from the context. Given a set of
vertices $S$, we use $G[S]$ to denote the subgraph induced by $S$, and $G
\setminus S$ to denote $G[V(G)-S]$. We denote
the complete graph on $n$ vertices as $K_n$ and the complete bipartite graph on
$n = n_1 + n_2$ vertices with partite sets of order $n_1, n_2$ as
$K_{n_1, n_2}$. As shorthand, we also refer to complete bipartite graphs as
\emph{bicliques}. We denote the neighbors of a vertex $u$ as $N(u)$. An edge
$(u, v)$ can be contracted by adding a vertex $uv$
with incident edges to the vertices $N(u) \cup N(v)$, and then deleting $u$
and $v$. We also define the contraction of a connected
subgraph $H$ as the iterative contraction of its edges (order does not matter
due to connectivity). For a set $S$ we denote its power set as $\mathcal{P}(S)$.

Current and prior D-Wave hardware layouts are based on the more general
\emph{Chimera graphs}. Visualized in Fig. \ref{figure:chimera}, a Chimera graph
\chimera~is an $M \times N$ grid of biclique $K_{L,L}$ cells. For example, the
latest D-Wave 2000Q hardware is based on a \chimeraargs{4}{16}{16} graph. In
the context of Chimera hardware, we assume that the qubits are labeled by their
location in the Chimera layout: $(\ell_r, \ell_c, \ell_p, \ell_h)$ where
$1 \leq \ell_r \leq M$ identifies a row, $1 \leq \ell_c \leq N$ identifies a
column, $\ell_p \in \{1, 2\}$ identifies a partite set, and
$1 \leq \ell_h \leq L$ denotes the in-cell height index.

\subsection{Minor Embedding for Adiabatic Quantum Programming}

\begin{figure}[t]
\centering
\includegraphics[width=\textwidth]{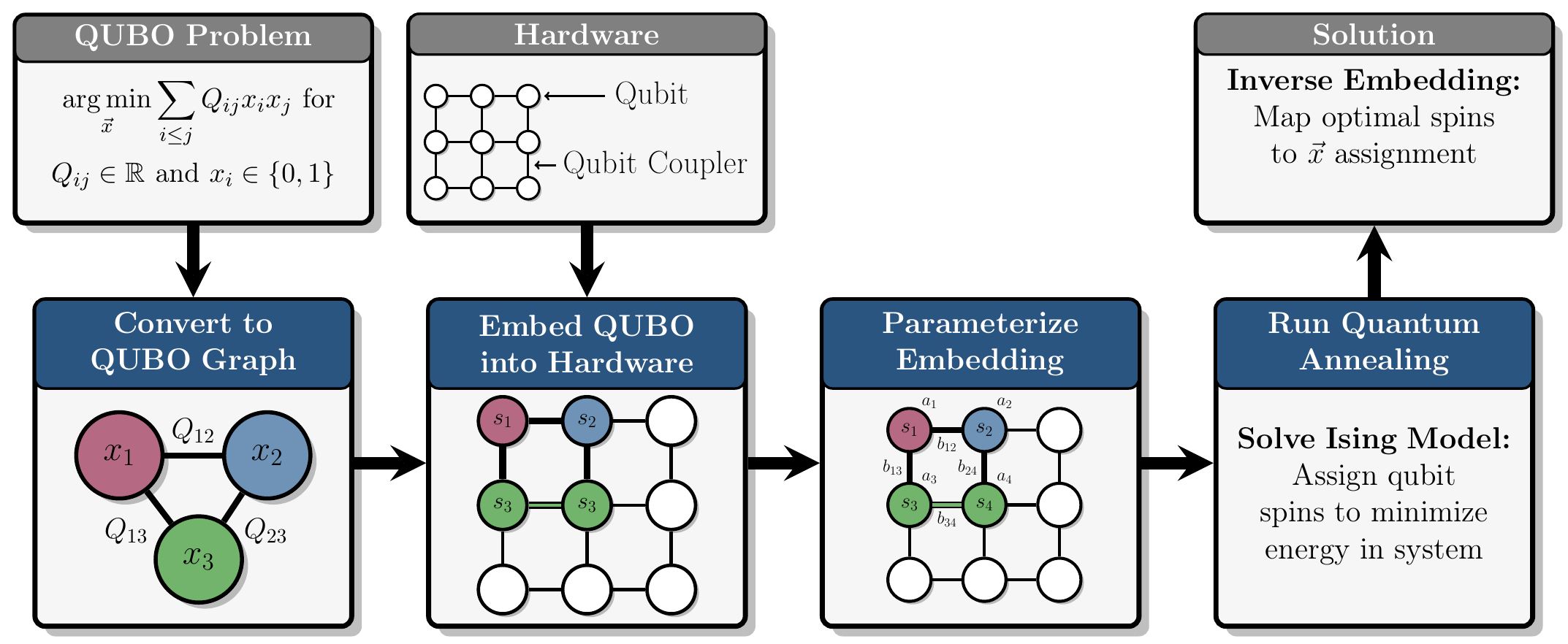}
\caption{A typical workflow for running QUBO-formulated optimization problems
on an AQC processing unit. Finding efficient and effective embedding algorithms
is an area of active research.}
\label{figure:pipeline}
\end{figure}

Programming a quantum annealer, such as the D-Wave hardware, requires setting
the parameters that define the underlying Ising Model. This process includes
defining the positive and negative spins as variable assignments such that
logical dependencies are maintained within the restricted connectivity of the
hardware graph. Recently, several efficient compilation methods have been
proposed for managing this process \cite{choi2008minor, humble2014integrated,
rieffel2015case, venturelli2015quantum}.

A generalized compilation pipeline is shown in Fig.~\ref{figure:pipeline}. A
common entry point into these compilation frameworks is the quadratic
unconstrained binary optimization (QUBO) problem. Given variables $x_1, \dots,
x_n$ where $x_i \in \{0, 1\}$ and constants $c_{ij} \in \mathbb{R}$, the
QUBO problem is to compute

\begin{equation*}
\label{eq:qubo}
    \displaystyle \argmin \sum_{i \leq j} c_{ij} x_i x_j.
\end{equation*}

QUBO has become a standard input format for quantum annealers, similar to the
linear program format used in efficient classical solvers such as CPLEX. Many
constrained optimization problems can be converted directly to QUBO form
\cite{boros2002pseudo}.

A QUBO can be converted directly into a graph $P$ with vertices $V(P) = \{x_1,
\dots, x_n\}$, edges $E(P) = \{(x_i, x_j) ~|~ i \neq j, c_{ij} \neq 0\}$,
vertex weights $c_{ii}$, and edge weights $c_{ij}$ for $i \neq j$. Viewing a
QUBO as a graph is particularly useful when selecting sets of physical qubits
to represent the QUBO variables, since this assignment is known as graph
\emph{minor embedding}:

\begin{definition}[Minor Embedding]
Given two graphs $P$ and $H$, a minor embedding of $P$ into $H$ is a function
$\phi : V(P) \to \mathcal{P}(V(H))$ that assigns each vertex in $P$ to a vertex
set from $V(H)$ such that the following properties hold:
\begin{enumerate}[noitemsep, topsep=3pt]
\item Vertex sets cannot overlap: $\phi(u) \cap \phi(v) = \emptyset$ for all
distinct $u, v \in V(P)$.
\item Vertex sets induce connected subgraphs: $H[\phi(u)]$ is connected for
every $u \in V(P)$.
\item Edges are represented: $(u, v) \in E(P) \to (u', v') \in E(H)$ for some
$u' \in \phi(u)$ and $v' \in \phi(v)$.
\end{enumerate}
\end{definition}

\noindent From a graph-theoretic perspective, this embedding defines the vertex
deletions and edge contractions necessary to find $P$ as a minor of $H$. From
the physics perspective, this embedding assigns an appropriate set of
\emph{physical qubits} to collectively represent a \emph{logical qubit}, and
QUBO weights are adjusted for this embedding by distributing each logical
qubit's weight over its vertex sets' physical qubits \cite{choi2008minor}.
Hence, compiling a QUBO into AQC hardware reduces to the problem of finding a
minor embedding.

The problem of finding a minor embedding is NP-hard for general graphs,
witnessed with a trivial reduction from \problemfont{SubgraphIsomorphism}.The
most famous minor-embedding result comes from the Robertson-Seymour graph
minor theory \cite{robertson1995graph}, which implies that there is a
polynomial-time algorithm for finding an embedding of a fixed problem graph
into any potential hardware graph. However, this algorithm assumes the size of
the problem graph is a constant and uses it exponentially, therefore the result
is not expected to yield practical embedding algorithms Choi notes that a
similar problem has been previously studied in parallel computing
\cite{leighton2014introduction} where a job needs to be distributed over a
cluster's nodes, but existing results are incompatible with the requirement of
a graph minor embedding \cite{choi2011minor}.

In addition to constructing a minor embedding, in practice we also want to
\emph{tune} the embedding to have beneficial graph properties. Finding an
embedding with a minimum hardware footprint, measured in qubits, would be
preferable to more wasteful embeddings. Experimental evidence also suggests
that large vertex sets lead to poor solutions in practice, so minimizing the
diameter of each vertex set's induced subgraph is desirable. Thus, in addition
to the NP-hard problem of generating a single embedding, we are also interested
in searching over the space of embeddings.

Examples of prior application-to-Ising-Model compilations include Lucas's
formulation of Karp's 21 NP-hard problems \cite{lucas2014ising}, NASA's rover
missions \cite{rieffel2015case, venturelli2015quantum}, applications in
computational chemistry \cite{kassal2011simulating}, and computer vision
\cite{neven2008image}.

\subsection{Related Work}
The notion of minor embedding QUBO problems into Chimera \chimera~hardware was
first introduced by Choi in 2008 \cite{choi2008minor}. Choi later
provided the first general purpose embedding algorithm \cite{choi2011minor},
\algorithmfont{TRIAD}, which embedded $K_{LN}$ (assuming $N \leq M$) into a
triangular portion of the D-Wave hardware. This embedding trivially provides
embeddings for all graphs of at least $LN$ vertices; however, no tuning
mechanism is provided to reduce the hardware footprint for problems with less
edges.

Klymko et al. \cite{klymko2014adiabatic} extended this work by providing an
alternative embedding algorithm for $K_{LN+1}$. While \algorithmfont{TRIAD}
could be extended for this extra vertex set, it unnaturally used all
remaining qubits. The embedding provided by Klymko et al. shifted these qubits
around such that all the vertex sets are (roughly) balanced. Klymko et al.
showed that this balanced embedded also proved resilient to hardware instances
with hard faults (missing qubits). Finally, the authors also introduced the
notion of QUBO rejection using structural graph properties. Specifically,
Klymko et al. showed that QUBOs with treewidth larger than $L(N+1)$ cannot be
embedded in \chimera. While treewidth is NP-complete to compute, Wang et al.
provided a linear-time approximation for problems based on Ollivier-Ricci
curvature \cite{wang2014ollivier}.

While the algorithms from \cite{choi2011minor} and \cite{klymko2014adiabatic}
ran in constant time and guaranteed an embedding, Cai et al.
\cite{cai2014practical} took a different approach by providing greedy
heuristics for embedding arbitrary QUBOs into arbitrary hardware graphs.
Experimental results provided by the authors show that, for very sparse graphs
such as 3-regular and grid graphs, the algorithm succeeded in embedding larger
QUBOs than previous embedding algorithms, while also using less qubits. This
so-called \emph{CMR algorithm} is the basis for the embedding algorithm
provided in the D-Wave API \cite{sapi}.

Most recently, Boothby et al. \cite{boothby2016fast} generalized the
\algorithmfont{TRIAD} embedding into a class of \emph{native clique embeddings}
for $K_{LN}$. They show that, unlike the \algorithmfont{TRIAD} embedding,
exponentially-many native clique embeddings exist in a given Chimera graph,
making it possible to construct one that avoids hard faults. Additionally, they
provide a polynomial-time dynamic programming algorithm for computing the
maximum native clique possible in a Chimera graph with hard faults.

\section{Virtual Hardware Framework}

At the core of our framework is a \emph{virtual hardware layer} created to
provide a cleaner interface for finding minor embeddings. The introduction of
this intermediary representation splits the minor embedding process into two
phases:

\begin{enumerate}
\item Find the initial embedding. Starting with a virtual hardware template
that allocates physical resources, find a virtual embedding function into the
virtual hardware.
\item Iteratively tune the embedding. After obtaining an initial embedding,
apply reduction routines to tune both the virtual embedding function and
virtual hardware to adjust physical hardware resource usage.
\end{enumerate}

Fig. \ref{figure:framework} illustrates this iteration. Provided with an initial
virtual hardware template $\virtualhardware$ and its embedding $\psi$ into the
physical hardware, a virtual embedding $\phi$ is sufficient for finding a valid
minor embedding of the QUBO into the physical hardware. By iterating reduction
subroutines, a sequence of improved embeddings
$(\phi, \virtualhardware, \psi) \to (\phi', \virtualhardware', \psi') \to (\phi'', \virtualhardware'', \psi'')$
each produce a full embedding with reduced hardware usage.

\begin{figure}[!h]
\centering
\includegraphics[width=\textwidth]{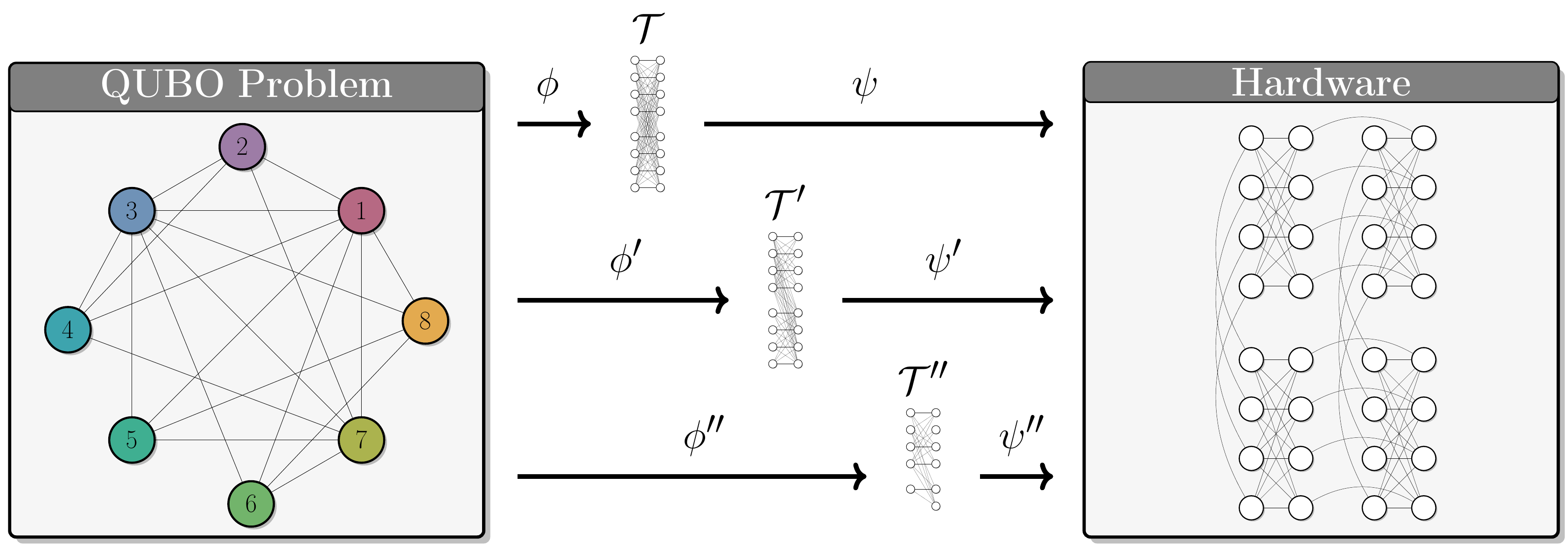}
\caption{A high-level overview of the virtual hardware framework, including the
iterative tuning of the virtual embedding $\phi$, virtual hardware template
$\virtualhardware$, and the physical embedding $\psi$.}
\label{figure:framework}
\end{figure}

\noindent Formally, we assume the problem is formulated as a graph $\program$
and the hardware layout as a graph $\hardware$. The \emph{virtual hardware
template} $\virtualhardware$ is a graph embeddable into $\hardware$. This
embedding denotes an allocation of qubits in $\hardware$ into virtual qubits in
$\virtualhardware$, encoded with a
\emph{physical embedding function} $\psi : V(\virtualhardware) \to \mathcal{P}(V(\program))$.
For bookkeeping, we require that each edge in $\virtualhardware$ represents
exactly one edge in $H$ -- therefore removing edges in the virtual hardware has
a corresponding meaning on the physical hardware footprint. Since we want to
define a virtual hardware template that scales with the physical hardware, we
define virtual hardware templates in terms of \emph{families}:

\begin{definition}[Chimera-Compatible Virtual Hardware Template Family]
A virtual hardware template family $\mathcal{F}$ is a set of virtual hardware graphs defined with a corresponding family of physical hardware
embedding functions $\Psi$, such that for all $L, M, N \in \mathbb{Z}^+$, there
exists $\psi \in \Psi$ and $\virtualhardware \in \mathcal{F}$
such that $\psi$ minor embeds $\virtualhardware$ into \chimera.
\end{definition}

Finding a \emph{virtual embedding function} $\phi : V(\program) \to
\mathcal{P}(V(\virtualhardware))$ is sufficient for finding the initial
embedding $\chi : V(P) \to \mathcal{P}(V(H))$, which can be constructed by
letting $\chi(u) = \bigcup_{x \in \phi(u)} \psi(x)$. We compute this virtual
embedding function with an \emph{embedding subroutine}:

\begin{definition}[Embedding Subroutine]
An embedding subroutine takes as input a problem graph $\program$ and virtual
hardware $\virtualhardware$, and outputs a virtual embedding
$\phi : V(\program) \to \mathcal{P}(V(\virtualhardware))$ or the keyword
\texttt{FAIL}.
\end{definition}

After finding a full minor embedding function, we then apply
\emph{reduction subroutines} to produce tuned embeddings:

\begin{definition}[Reduction Subroutine]
A reduction subroutine takes as input a problem $P$, virtual hardware
$\virtualhardware$, and virtual embedding $\phi$, then outputs an updated virtual hardware $\virtualhardware'$ and virtual embedding $\phi'$ (potentially
identical to $\virtualhardware$ and $\phi$).
\end{definition}

After reduction subroutines are applied, an updated physical embedding function
$\psi'$ can be recovered from the original $\psi$ and the final virtual
hardware $\virtualhardware'$ by using only the physical qubits needed to
represent the edges in $\virtualhardware'$. Again, we have a full embedding
$\chi'$ of the problem into the physical hardware by combining $\phi'$ and
$\psi'$.

\subsection{Biclique Virtual Hardware}
We now present an implementation of this framework using a
\emph{biclique virtual hardware}, an embedding subroutine based on both Choi's
TRIAD and Klymko et al.'s embedding algorithm, and provide two reduction
subroutines for minimizing the total number of qubits used. We start with the
virtual hardware:

\begin{definition}[Biclique Virtual Hardware Template Family]
A \chimera~ hardware contains a biclique $K_{LM, LN}$ virtual hardware $\virtualhardware$ with partite sets
$L(\virtualhardware) = \{v_1, \dots, v_{LM}\}$ and
$R(\virtualhardware) = \{h_1, \dots, h_{LN}\}$; we refer to these as the
\emph{vertical} and \emph{horizontal} partite sets, respectively.
The embedding function defining the minor embedding is given by
\begin{align*}
\psi(v_i) &= \{(j, \lceil i / L \rceil, 1, i \bmod L) ~|~ 1 \leq j \leq M\},
\text{ and }\\
\psi(h_i) &= \{(\lceil i / L \rceil, j, 2, i \bmod L) ~|~ 1 \leq j \leq N\}.
\end{align*}
\end{definition}

\begin{figure}[t]
\centering
\includegraphics[width=\textwidth]{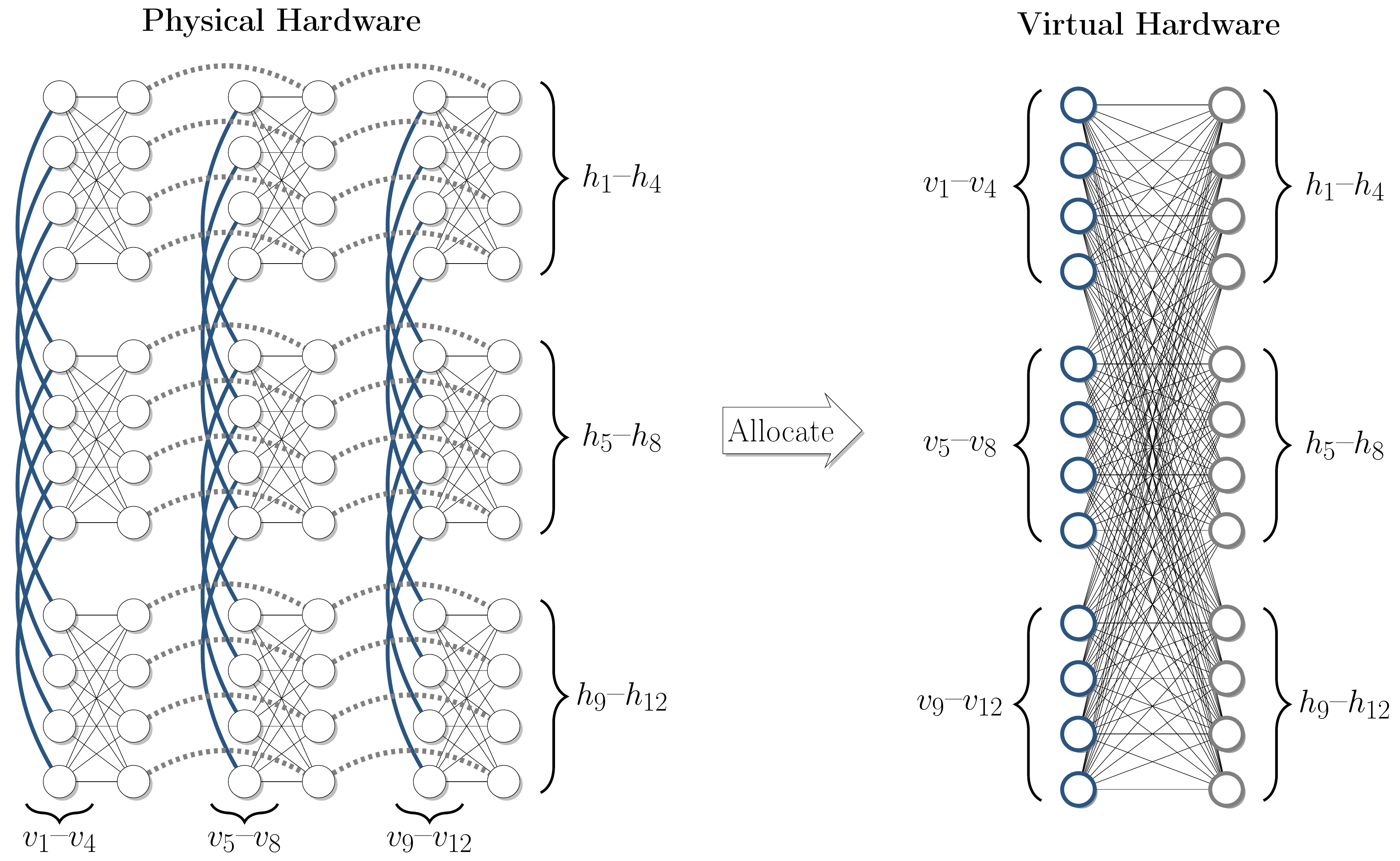}
\caption{The $K_{12,12}$ biclique virtual hardware for Chimera(4, 3, 3). Thick
blue edges show allocations to vertical vertex sets, and dashed gray edges show
the horizontal vertex set allocations.}
\label{figure:virtual_hardware}
\end{figure}

The intuition behind this allocation is a
partitioning of the edges in Chimera graphs (c.f. Fig.
\ref{figure:virtual_hardware}). There are three such edge types --
intra-cell, vertical inter-cell, and horizontal inter-cell -- and the
inter-cell edges provide the highest connectivity increase per minor
contraction. Therefore allocating maximal vertical and horizontal paths
provides a virtual hardware with relatively large degree per vertex.

The biclique virtual hardware is fairly robust to physical hardware
specifications by not requiring a square Chimera grid like previous algorithms,
nor depending on the fact that $L=4$ in existing hardware implementations. A
biclique virtual hardware can also be allocated from a hardware implementation
with hard faults; however, in the naive allocation we find that each missing
qubit removes a full vertical or horizontal path. Managing hardware
implementations with hard faults is less a concern than in prior work, with
more mature hardware yields and the introduction of software post-processing
methods for emulating missing qubits (e.g. the Full-Yield Chimera Solver
provided in D-Wave SAPI 2.4
\cite{sapi}).

\subsection{Biclique Embedding and Reduction Subroutines}
In this subsection we develop a baseline set of embedding and reduction
subroutines utilizing the biclique virtual hardware template. We start by
providing an embedding for a complete graph on $\min(LM, LN)$ vertices. At a
high level the embedding assignment is straightforward: a single virtual qubit
in the vertical partite has edges to every virtual qubit in the horizontal
partite, and vice-versa. Therefore, to ensure that every two problem vertices
are joined by an edge, we map each problem vertex to a pair of virtual qubits:

\begin{subroutine}[\algorithmfont{Native-Embed}]
Given a problem graph $\program$ with $V(\program) = \{u_1, \dots, u_n\}$ where
$n \leq \min(LM, LN)$ and a biclique virtual hardware $\virtualhardware$ with
partites $L(\virtualhardware) = \{v_1, \dots, v_{LM}\}$ and
$R(\virtualhardware) = \{h_1, \dots, h_{LN}\}$, \algorithmfont{Native-Embed}
produces an embedding $\phi$ by mapping $\phi(u_i) = \{v_i, h_i\}$ for
$1 \leq i \leq n$.
\end{subroutine}

As defined, \algorithmfont{Native-Embed} redundantly has two edges between
every pair of vertex sets $\phi(u_i)$ and $\phi(u_j)$, for $i \neq j$; namely, the
edges $(v_i, h_j)$ and $(v_j, h_i)$. Recall that we defined $\psi$ such that
each edge in $\virtualhardware$ represents a unique edge in $\hardware$, so
this redundancy in the virtual hardware represents an actual redundancy in the
physical embedding. To gauge the wastefulness, we \emph{score} a virtual
hardware and its virtual embedding:

\begin{subroutine}[\algorithmfont{Qubit-Scoring}]
Suppose we are given standard input $\program$, $\virtualhardware$, and $\phi$.
For each virtual qubit $v_i \in L(\virtualhardware)$, let
$I_{v_i} = \{j ~|~ (v_i, h_j) \in E(\virtualhardware)\}$ be
its \emph{index set} -- the range of neighbors it has on the virtual hardware.
Define the score for each left partite vertex as

\begin{align*}
\text{\algorithmfont{score}}(v_i) = 1 + \left\lfloor \frac{max(I_{v_i})}{L}
\right\rfloor - \left\lfloor \frac{min(I_{v_i})}{L} \right\rfloor.
\end{align*}

Each $h_i$ is assigned an index set and score analogously.
Then the qubit score for $\phi$ and $\virtualhardware$ is
\begin{align*}
\sum_{v_i \in L(\virtualhardware)} \text{\algorithmfont{score}}(v_i) + \
\sum_{h_i \in R(\virtualhardware)} \text{\algorithmfont{score}}(h_i).
\end{align*}
\end{subroutine}

At a high level, \algorithmfont{Qubit-Scoring} computes the number of physical
qubits that must be used with the current virtual hardware and virtual
embedding. If removing a redundant edge reduces the score, then we have also
reduced physical hardware usage. If removing a redundant edge does \emph{not}
affect the score, then we know that this particular edge is not requiring extra
hardware usage by itself; however, a sequence of non-score-reducing redundant
edge removals could potentially reduce the score. Therefore, it is non-trivial
to identify which of the redundant edges should be removed for optimal hardware
resource minimization.

Based on this observation, we provide two evaluation methods for computing
virtual hardware minimization. First, \algorithmfont{Qubit-Evaluation} computes
all possible redundant edge removals and chooses the one with minimum score;
this calculation is exponential in the number of redundant edges. A faster
evaluation method \algorithmfont{Fast-Qubit-Evaluation} greedily keeps the
lexicographically-first edge, providing a \emph{minimal} (but not necessarily
minimum) score in linear time.

\begin{subroutine}[\algorithmfont{Qubit-Evaluation}]
Suppose we are given standard input $\program$, $\virtualhardware$, and $\phi$. Let $S$ be the set of problem vertices mapped to at
least one virtual qubit on each partite, let $\mathcal{E}$ be the set of all
edge sets $E$ on the virtual hardware such that for each $u, v \in S$, there is
exactly one edge $(u', v') \in E$ with $u' \in \phi(u)$ and $v' \in \phi(v)$.
Then \algorithmfont{Qubit-Evaluation} returns
$\argmin_{E \in \mathcal{E}} \text{\algorithmfont{Qubit-Scoring}(E)}$.
\end{subroutine}

\begin{subroutine}[\algorithmfont{Fast-Qubit-Evaluation}]
Suppose we are given standard input $\program$, $\virtualhardware$, and $\phi$.
Let $S$ be the set of problem vertices mapped to at
least one virtual qubit on each partite. Then
\algorithmfont{Fast-Qubit-Evaluation} returns
$E = \{(v_i, h_j) ~|~ i \leq j,  ~v_i \in \phi(x), ~h_j \in \phi(y) \text{ for } x,y \in V(P) \text{ and } y \in N(x)\}$.
\end{subroutine}

The last step is to use this reduced edge set to construct a reduced virtual
hardware, computed using \algorithmfont{Qubit-Reduce}:

\begin{figure}[h!]
\centering
\includegraphics[width=\textwidth]{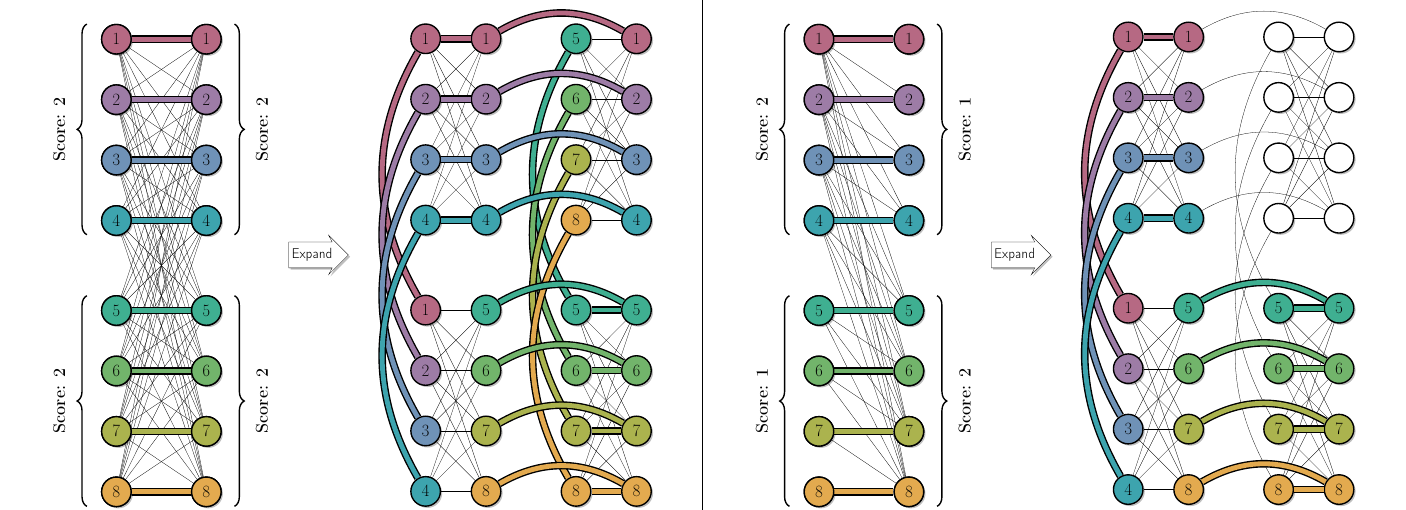}
\caption{(Left) The \algorithmfont{Native-Embed} embedding with ``+''-shaped
vertex sets; (Right) The embedding reduced by \algorithmfont{Qubit-Reduce},
with ``L''-shaped vertex sets.}
\label{figure:vhscoring}
\end{figure}

\begin{subroutine}[\algorithmfont{Qubit-Reduce}]
Given standard input $\program$, $\virtualhardware$, $\phi$ and an evaluation
subroutine, \algorithmfont{Qubit-Reduce} computes a set of redundant edges to
be removed $E$, and outputs the current virtual embedding $\phi$ and a new
virtual hardware $\virtualhardware'$ with vertices $V(\virtualhardware)$ and
edges $E(\virtualhardware) - E$.
\end{subroutine}

While fairly simple, \algorithmfont{Qubit-Reduce} has the potential to reduce
qubit usage by 50\%. This ratio occurs when \algorithmfont{Native-Embed}'s
``+''-shaped vertex sets on the physical hardware are reduced to ``L''-shaped
vertex sets (as described by Boothby et al. \cite{boothby2016fast}). Fig.
\ref{figure:vhscoring} visualizes this reduction.

Up to this point, we have implicitly assumed that the problem graph
was complete (i.e. we needed to enforce every edge). However, we can achieve
further hardware resource reduction by assuming that the problem is missing
edges. Specifically, by shuffling the assignment of vertex sets on the biclique
virtual hardware, we can group together those vertices with edges between them,
resulting in shorter vertex sets. This computation can be done with a scheme of
subroutines, \algorithmfont{$k$Exchange-Reduce}. In local search terminology,
we compute a deterministic gradient descent on the $k$-exchange neighborhood
without restarts.

\begin{subroutine}[\algorithmfont{$k$Exchange-Reduce}]
Given standard input $P, \virtualhardware, \phi$, and \emph{neighborhood
exchange parameter} $k \geq 2$,
the subroutine \algorithmfont{$k$Exchange-Reduce} computes a new virtual
embedding $\phi'$ with the following steps:
\begin{enumerate}
\item Let $\phi' = \phi$.
\item Starting from $\phi'$, compute all $\binom{n}{k}k!$ ways to reassign
exactly $k$ problem vertices in each partite, and score each qubit
reassignment. (For example, if $\phi(u_1) = \{v_1, h_1\}$ and
$\phi(u_2) = \{v_2, h_2\}$, then their 2-exchange on the left partite is
$\phi(u_1) = \{v_2, h_1\}$ and $\phi(u_2) = \{v_1, h_2\}$).
\item Let $\phi'$ be the reassignment with the lowest score.
\item Repeat until no $k$-exchange leads to a score reduction, and return
$\phi'$ and $\virtualhardware$.
\end{enumerate}
\end{subroutine}

For a fixed $k$, run time for \algorithmfont{$k$Exchange-Reduce} is
$\binom{n}{k}k! = O(n^k)$ per iteration with a maximum of $L^2(M+N)$
iterations. With the standard assumptions that $L$ is a constant and
$\sqrt{n} = \max(M,N)$, \algorithmfont{$k$Exchange-Reduce} has a run time of
$O(n^{k+1/2})$.

\subsection{Emulation and Enhancement}
Applying the tools introduced in the last subsection, we can emulate Choi's
$K_{LN}$ \algorithmfont{TRIAD} algorithm \cite{choi2011minor} with \algorithmfont{Native-Embed} and
\algorithmfont{Qubit-Reduce}. Klymko et al.'s $K_{LN+1}$ embedding
\cite{klymko2014adiabatic} can be found by tweaking
\algorithmfont{Native-Embed} to embed $u_1, \dots, u_{LN-1}$ as usual, but also
setting $\phi(u_{LN}) = \{v_{LN}\}$ and $\phi(u_{LN+1}) = \{h_{LN}\}$. We note
that doing so forces the first $LN-1$ vertex sets to cover both the vertical and
horizontal partites in order to be adjacent to the last two vertices, therefore
applying \algorithmfont{Qubit-Reduce} has a limited effect and is not
recommended for general use.

One advantage of emulating existing algorithms in this framework is for the
application of virtual hardware-specific reduction subroutines; namely,
\algorithmfont{Qubit-Reduce} and \algorithmfont{$k$Ex-Reduce}. In Section
\ref{section:comparing_oct} we see that the subroutines do in fact produce
smaller embeddings without unreasonably increasing run times.

\subsection{Summary}
In this section we defined a biclique virtual hardware formed naturally from
the Chimera graph by exploiting its grid-like structure and high intra-cell
connectivity.

We also defined a full baseline stack of embedding and reduction subroutines.
As noted in the last subsection, this framework is sufficient for emulating the
best existing algorithms for dense problem graphs in hardware layouts without
faults. Furthermore, we can apply additional reduction routines to achieve
reduced embedding footprints. In total, these results serve as a full
proof-of-concept motivating the use of virtual hardware and the development of
specialized and modular subroutines. In the next section we take the next step
and move beyond existing embedding algorithms by exploiting the bipartite
structure in problem graphs to tackle larger, more sparse problems.

\section{Utilizing Bipartite Problem Structure}
In the last section we emphasized the structural properties of the Chimera
hardware graph, deriving the biclique virtual hardware and its subroutines. In
this section we utilize bipartite structure from the problem graph.
Specifically, we use the notion of \emph{odd cycle transversals} to decompose
problem graphs and extract a maximal bipartite induced subgraph. We start by
defining the odd cycle transversal and its limitations in the Chimera graph,
then describe an initial embedding subroutine \algorithmfont{OCT-Embed}, and
finally propose a faster heuristic, \algorithmfont{Fast-OCT-Embed}.

\subsection{Odd Cycle Transversal}
One metric for gauging the ``bipartite-ness'' of a graph $G$ is the smallest
set of vertices preventing $G$ from being bipartite, a minimum odd cycle
transversal:

\begin{definition}[Odd Cycle Transversal (OCT)]
The \emph{odd cycle transversal} of a graph $G$ is a set of vertices $S$ such
that $G\setminus S$ is a bipartite graph. We denote the size of a minimum OCT
set as the OCT number, $OCT(G)$, and the problem of computing $OCT(G)$ as
\problemfont{MinOCT}.
\end{definition}

Unfortunately, \problemfont{MinOCT} is NP-hard \cite{lewis1980node} and does
not have a constant factor approximation algorithm unless P = NP
\cite{lund1993approximation}. However, the problem is fixed-parameter tractable
(FPT) when parameterized by the natural parameter (solution size). In other
words, graphs with small OCT numbers will also have quickly-computable OCT
numbers, regardless of total graph size. Given that the biclique virtual
hardware is most efficiently utilized when embedding problem graphs with small
OCT numbers, we expect embeddable problem graphs will have an efficiently
computable OCT decomposition. As a baseline we use Reed et al.'s $O(3^{k}kmn)$
algorithm for computing solutions of size $k$, which is known to have several
simplifications and optimizations \cite{lokshtanov2009simpler,
huffner2005algorithm}. Other algorithms for specialized instances also exist
\cite{agarwal2005approximation, lokshtanov2012subexponential}.

We note that \problemfont{MinOCT} and the problem of computing the
size of the maximum bipartite induced subgraph (denoted by
\problemfont{MaxBipartite}) are complements, in the sense that an exact
solution to one problem also provides a solution to the other. However, an
\emph{approximation} for one problem is not an approximation for the other, so
some care must be taken when choosing which problem to approximate.

\subsection{OCT and the Chimera Graph}
In prior work, Klymko et al. showed that the Chimera graph has treewidth
bounded by $O(LN)$, assuming $N \leq M$ \cite{klymko2014adiabatic}. In this
section we show that the maximum $OCT(G)$ over all Chimera-embeddable graphs
$G$ is bounded by all three Chimera parameters. First, we note that treewidth
and OCT describe different graph structure:

\begin{proposition}
The treewidth of a graph is independent of its OCT number.
\end{proposition}

\begin{proof}
Consider two families of graphs:
\begin{enumerate}
\item The class of grid graphs. These graphs are known to have treewidth
proportional to the smallest grid dimension \cite{diestel2005graph}, but have
an OCT number of 0 since they are bipartite.
\item The class of trees with their leaves replaced with triangles. These
graphs have treewidth at most three (a tree decomposition exists where each bag
contains at most a triangle and its neighbor in the tree), but unbounded OCT
number since each (disjoint) triangle contains at least one OCT vertex.
\end{enumerate}
We have shown that one property cannot bound the other, therefore they are
independent. \smartqed \qed
\end{proof}

With this independence established, we proceed to show upper and lower bounds
on the maximum OCT number of a Chimera-embeddable graph.

\begin{lemma}
$OCT(G) \leq \min(|L(B)|, |R(B)|)$ for all minors $G$ of a bipartite graph $B$
with vertex partite sets $L(B)$ and $R(B)$.
\end{lemma}
\begin{proof}
Let $\phi$ be a minor embedding of $G$ into $B$, and without loss of generality,
let $L(B)$ be the smaller of the two partite sets. Let $S = \{x ~|~ x \in V(G)
\text{ and } u \in \phi(x) \text{ for } u \in L(B)\}$, then we know that $|S|
\leq |L(B)|$. $V(G) - S$ is necessarily bipartite since $\phi(x)$ is composed
of vertices from $R(B)$ for $x \in V(G)-S$, therefore
$OCT(G) \leq |S| \leq |L(B)|$. \smartqed \qed

\end{proof}
\begin{corollary}
$OCT(G) \leq LMN$ for all Chimera-embeddable graphs $G$.
\end{corollary}

\begin{lemma}
\label{lemma:octlowerbound}
There exists a Chimera-embeddable graph $G$ such that $OCT(G) \geq (L-1)MN$.
\end{lemma}
\begin{proof}
We construct $G$ by contracting $L-1$ vertex-disjoint edges in each cell of a
Chimera graph. Each cell is now a $K_{L+1}$ clique and $L-1$ of these vertices
must be included in an OCT set, therefore $OCT(G) \geq (L-1)MN$.
\smartqed \qed
\end{proof}

While the treewidth of Chimera graphs only grows in two dimensions ($L$ and
$\min(M, N)$), Lemma \ref{lemma:octlowerbound} shows that the minimum odd cycle
transversal will increase if $L$, $M$, or $N$ is increased. Therefore we
recommend using a \emph{minimal} odd cycle transversal as a proxy for
estimating how much hardware a problem graph's embedding will require. A
minimal OCT set is fast to compute and reflects more of the actual hardware
usage than treewidth.

In addition to gauging how much hardware a problem's embedding will require, we
can also use the minimum odd cycle transversal to recognize when certain
problems are \emph{not} embeddable. The Klymko et al.
\cite{klymko2014adiabatic} result shows that problems with treewidth larger
than $(L+1)N$ cannot be embedded, but this bound does not apply to classes of
small treewidth, such as series-parallel graphs \cite{eppstein1992parallel}.
However, the OCT rejection criterion provides a characterization of
unembeddable graphs in terms of odd cycles, therefore it encompasses a different class of graphs (including series-parallel graphs, which can have an unbounded number of odd cycles).

\subsection{Computing OCT and \algorithmfont{OCT-Embed}}
\label{section:kernelization_and_oct-embed}
As mentioned previously, the fastest-known algorithm for computing the OCT
number is exponential in the solution size, so we want to prune graphs if
possible. One method of doing that is by removing tree-like structure:

\begin{proposition}
To compute \problemfont{MinOCT} on a graph $G$, it is sufficient to compute
\problemfont{MinOCT} on the maximal 2-edge-connected subgraphs of $G$.
\end{proposition}

\begin{proof}
We induct on the number of 2-edge-connected maximal subgraphs. If there is no
such subgraph, then every edge is a bridge and the graph is a tree, therefore
the claim is true.

Suppose instead that there are $k$ such subgraphs in $G$ and the claim is true
for all graphs with $k-1$ such subgraphs. We can decompose the graph into
maximal 2-edge-connected subgraphs by computing a chain decomposition
\cite{schmidt2013simple} on $G$ to identify its \emph{bridges}. Removing these
bridges produces each maximal 2-edge-connected subgraph as a connected
component. Further, contracting these subgraphs creates a tree with the
contracted subgraphs as vertices and the bridges as the edges. Pick a subgraph
$S$ that is a leaf on this tree, and let $(v_1, v_2)$ be the bridge separating
$S$ from $G \setminus S$. By the induction hypothesis, we can compute $OCT(G
\setminus S)$ and $OCT(S)$ on their maximal 2-edge-connected subgraphs,
therefore all that remains is to show that these two partial solutions are
compatible.

Suppose that the partial solutions are expressed as a coloring: vertices in the
left partite set are colored $L$, the right partite set $R$, and neither
partite set (e.g. in the OCT set) as $N$. If at least one of $v_1, v_2$ is
colored $N$, or if one is colored $L$ and the other $R$, then these partial
solutions are compatible as-is. Suppose to the contrary that both are colored
with the same partite set color. Then in $S$ we recolor $L \to R$ and
$R \to L$. This recoloring does not change $OCT(S)$, and the partial solutions
are now compatible.
\smartqed \qed
\end{proof}

This preprocessing step is fast, costing only an additive $O(m)$ run time when
using Schmidt's chain decomposition algorithm \cite{schmidt2013simple}. This
approach also provides an opportunity for parallelization if the graph has many
2-edge-connected maximal subgraphs. While this technique applies to any graph,
we can take advantage of the 2-edge-connectivity in the class of
series-parallel graphs by exploiting \emph{nested ear decompositions}:

\begin{proposition}
\label{proposition:seriesparallel}
$OCT(G)$ can be computed in linear time for a series-parallel graph $G$.
\end{proposition}
\begin{proof}
The proof of Proposition \ref{proposition:seriesparallel} can be found in
Appendix \ref{appendix:eppstein}.
\smartqed \qed
\end{proof}

We conclude this subsection by defining an embedding subroutine that uses an
OCT-decomposition to embed into the biclique virtual hardware. At a high level,
\algorithmfont{OCT-Embed} first computes a minimum OCT set, embeds the OCT
vertices as if they were a complete graph, and then embeds the bipartite
induced subgraph directly into the biclique virtual hardware (Fig.
\ref{figure:octembed}.

\begin{subroutine}[\algorithmfont{OCT-Embed}]
Let $P$ be a problem graph with $V(P) = S \cup L \cup R$, where
$S = \{u_1, \dots, u_i\}$ is a minimum OCT set, and
$L = \{u_{i+1}, \dots, u_j\}$ and $R = \{u_{j+1}, \dots, u_n\}$ are a maximum
bipartite induced subgraph. Let $\virtualhardware$ be a biclique virtual hardware with partites $L(\virtualhardware) = \{v_1, \dots, v_{LM}\}$ and $R(\virtualhardware) = \{h_1, \dots, h_{LN}\}$. If $j \leq LM$ and $n-i \leq LN$, then \algorithmfont{OCT-Embed} produces an
embedding $\phi$ by mapping:
\begin{align*}
\phi(u_x) &=
\left\{
\begin{array}{ll}
\{v_x, h_x\} \text{ if } u_x \in S\\
\{v_x\} \text{ if } u_x \in L\\
\{h_{x-i}\} \text{ if } u_x \in R
\end{array}
\right.
\end{align*}
otherwise it outputs \texttt{FAIL}.
\end{subroutine}

\begin{figure}[!h]
\centering
\includegraphics[width=0.9\textwidth]{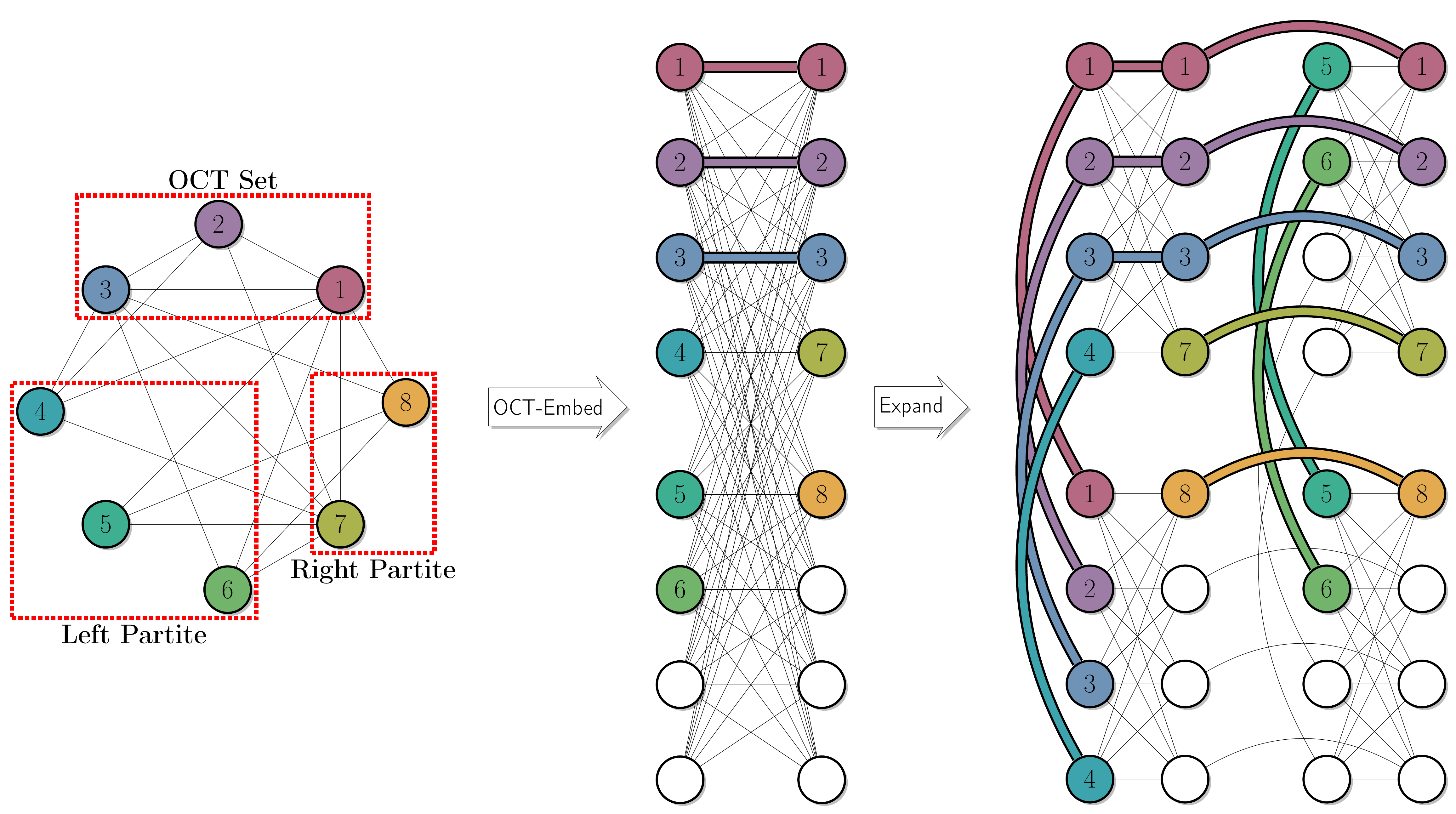}
\caption{Embedding an 8-vertex problem into \chimeraargs{4}{2}{2} using the
\algorithmfont{OCT-Embed} subroutine. For figure readability, vertex $u_i$ is
labeled with $i$.}
\label{figure:octembed}
\end{figure}

\subsection{Approximating OCT and \algorithmfont{Fast-OCT-Embed}}
\label{section:approximating_oct}
A downside to \algorithmfont{OCT-Embed} is its exponential run time,
restricting the subroutine's real-world applicability. However, an exact
solution to \problemfont{MinOCT} is not always required for a full embedding --
any odd cycle transversal decomposition will work as long as it fits into the
biclique virtual hardware. We utilize this fact to develop an approximation
algorithm for \problemfont{MaxBipartite}, and use this approximation algorithm
for two purposes: (1) as an initial solution for the iterative compression in
our algorithm for \algorithmfont{OCT-Embed}, and (2) as a standalone embedding
subroutine \algorithmfont{Fast-OCT-Embed}.

We approximate \problemfont{MaxBipartite} instead of \problemfont{MinOCT} for
two reasons. First, the Reed et al. algorithm \cite{reed2004finding} we use to
compute the exact OCT number uses a technique called \emph{iterative
compression}, where a solution of size $k+1$ is compressed to size $k$ over
several subgraph iterations. We can reduce the number of these iterations by
providing the algorithm with a large initial subgraph with at most $k$ OCT
vertices, therefore we have motivation for estimating a maximal bipartite
subgraph. Second, if we approximate \problemfont{MaxBipartite}, then our worst
approximations (in terms of magnitude) are when the graph has a large bipartite
graph. However, this implies a small OCT set, therefore the exact algorithm
will have an exponentially faster run time. Therefore approximating
\problemfont{MaxBipartite} makes more sense in this context.

Our approximation algorithm is outlined in Algorithm
\ref{algorithm:greedybipartite}. Partially motivated by the success of using a
greedy algorithm to compute exact solutions on series-parallel graphs, we found
that a minimum-degree--greedy algorithm also performed well in practice on
general graphs (c.f. Section \ref{section:comparing_oct}). In total, the
algorithm has a run time of $O(m)$ using a modification of Bataglj and
Zaver\u{s}nik's algorithm for computing $k$-core decompositions
\cite{batagelj2003algorithm}.

\begin{algorithm}[ht]
\caption{Greedy Maximal Bipartite Induced Subgraph}
\begin{algorithmic}[1]
\Function{GreedyBipartite}{$G$}
\State $L \gets \Callsf{GreedyIndSet}{G}$
\State $R \gets \Callsf{GreedyIndSet}{G \setminus L}$
\State \return{$L \cup R$}
\EndFunction
\State
\Function{\textsf{GreedyIndSet}}{$G$}
\State $S \gets \emptyset$
\While{$G$ not empty}
\State $v \gets \argmin_{u \in V(G)} d(u)$ \Comment{Pick any min degree vertex}
\State $S \gets S \cup \{v\}$
\State $G \gets G \setminus (\{v\} \cup \Call{N}v)$
\EndWhile
\State \return{$S$}
\EndFunction
\end{algorithmic}
\label{algorithm:greedybipartite}
\end{algorithm}

We begin the approximation factor analysis by noting that an approximation
algorithm for minimum independent set translates to \problemfont{MaxBipartite}:

\begin{lemma}
\algorithmfont{GreedyBipartite} implemented with an $\alpha$-approximation
\algorithmfont{GreedyIndSet} algorithm is an $\alpha$-approximation algorithm.
\label{lemma:indsetoct}
\end{lemma}

\begin{proof}
Let $S$ be a fixed set of vertices such that $G[S]$ is the larger partite of a
maximum bipartite induced subgraph.
We want to show that for every vertex \algorithmfont{GreedyBipartite} adds to
its solution $S'$, at most $\alpha$ vertices from $S$ are not chosen.
Let $L$ and $R$ be the first and second independent sets constructed by
\algorithmfont{GreedyIndSet}, respectively.
First, the set $L$ is chosen without (immediately) disqualifying any vertex in
$G \setminus L$ from being in $R$, so no vertices are disqualified from $S'$ in this step.
When constructing $R$, at most $\alpha$ vertices from $S$ are disqualified for
every vertex added to $R$, by definition of the approximation factor.
Therefore $R$ itself is an $\alpha$-approximation for the partite and a
$2\alpha$-approximation for \problemfont{MaxBipartite}.
If $|L| \geq |R|$ then we have shown at least an $\alpha$-approximation.
To show the approximation factor still holds when $|L| < |R|$, we want to show
that $L$ is a $\alpha$-approximation for \problemfont{MaxIndSet} in
$G \setminus R$. But the previous argument still holds, since at most $\alpha$
vertices from $S$ are disqualified from $S'$ for every vertex chosen from
$G \setminus R$. Therefore in both cases we have a $alpha$-approximation for
the larger partite of a maximum induced bipartite subgraph, therefore we have an
$\alpha$-approximation for \problemfont{MaxBipartite}.
\smartqed \qed
\end{proof}

\begin{corollary}
\algorithmfont{GreedyBipartite} is a $\frac{\Delta+2}{3}$-approximation and a
$\frac{2\bar{d}+3}{5}$-approximation for graphs with maximum degree $\Delta$
and average degree $\bar{d}$.
\end{corollary}

\begin{proof}
Halld\'{o}rsson and Radhakrishnan \cite{halldorsson1994greed} show that
\algorithmfont{GreedyIndSet} is a $\frac{\Delta+2}{3}$-approximation and a
$\frac{2\bar{d}+3}{5}$-approximation for maximum independent set. By Lemma
\ref{lemma:indsetoct}, the same approximation factors hold for
\algorithmfont{GreedyBipartite}.
\smartqed \qed
\end{proof}

\begin{corollary}
\algorithmfont{GreedyBipartite} is a $d$-approximation for $d$-degenerate
graphs.
\end{corollary}

\begin{proof}
We first want to show that \algorithmfont{GreedyIndSet} is a $d$-approximation,
this proof mirrors that of Lemma \ref{lemma:indsetoct}. Fix a maximum
independent set $S$. In each step of \algorithmfont{GreedyIndSet}, a vertex
added to the solution disqualifies at most $d$ vertices from $S$. Therefore
\algorithmfont{GreedyIndSet} is a $d$-approximation for a maximum independent
set, and applying Lemma \ref{lemma:indsetoct} shows that
\algorithmfont{GreedyBipartite} is a $d$-approximation for
\problemfont{MaxBipartite}.
\smartqed \qed
\end{proof}

Up to this point we have not assumed anything about the OCT set when computing
an approximation factor. However, as graphs get more dense the OCT set must also grow. We can show this by using \emph{degeneracy} as a metric for density:

\begin{definition}[Graph Degeneracy]
The \emph{degeneracy} of a graph $G$ is the smallest $k$ such that every subgraph of $G$ has a vertex of degree at most $k$.
\end{definition}

\begin{lemma}
\algorithmfont{GreedyBipartite} is a $(n-d)$-approximation for a $d$-degenerate
graph when $d \geq \frac{n}{2}$.
\label{lemma:GreedyBipartite}
\end{lemma}
\begin{proof}

When $d \leq \frac{n}{2}$, the desired graph can always be found as a subset of
$K_{n/2,n/2}$. However, for larger values of $d$, vertices must be moved from
the bipartite graph into the OCT set, specifically two vertices per additional
unit of degeneracy. This fact means that a $d$-degenerate graph can have at
most a bipartite subgraph on $2(n-d)$ vertices. Solving for the approximation
factor: $\alpha \cdot 2(n-d) = \frac{2n}{d}$, so
$\alpha = \frac{2n}{2d(n-d)} = \frac{n}{n-d} \cdot \frac{1}{d} \geq \frac{n}{n-d} \cdot \frac{1}{n} = \frac{1}{n-d}$.
\smartqed \qed
\end{proof}

\begin{proposition}
\label{proposition:approximation}
\algorithmfont{Fast-OCT-Embed} is a $\min(d, n-d)$-approximation for
$d$-degenerate graphs.
\end{proposition}
\begin{proof}
This result follows directly from Lemmas \ref{lemma:indsetoct} and
\ref{lemma:GreedyBipartite}.
\smartqed \qed
\end{proof}

In other words, the degeneracy-based approximation factor is best on very
sparse and very dense graphs. Swapping the approximation algorithm into our
embedding subroutine, we now define \algorithmfont{Fast-OCT-Embed}:

\begin{subroutine}[\algorithmfont{Fast-OCT-Embed}]
Let $P$ be a problem graph with $V(P) = S \cup L \cup R$, where
$S = \{u_1, \dots, u_i\}$ is an OCT set, and
$L = \{u_{i+1}, \dots, u_j\}$ and $R = \{u_{j+1}, \dots, u_n\}$ are a maximum
bipartite induced subgraph. Let $\virtualhardware$ be a biclique virtual hardware with partites $L(\virtualhardware) = \{v_1, \dots, v_{LM}\}$ and $R(\virtualhardware) = \{h_1, \dots, h_{LN}\}$. If $j \leq LM$ and $n-i \leq LN$, then \algorithmfont{OCT-Embed} produces an
embedding $\phi$ by mapping:
\begin{align*}
\phi(u_x) &=
\left\{
\begin{array}{ll}
\{v_x, h_x\} \text{ if } u_x \in S\\
\{v_x\} \text{ if } u_x \in L\\
\{h_{x-i}\} \text{ if } u_x \in R
\end{array}
\right.
\end{align*}
otherwise it outputs \texttt{FAIL}.
\end{subroutine}

\subsection{Summary}

In summary, the odd cycle transversal provides a structured method for
decomposing problems and embedding them smartly into the Chimera hardware. We
showed that OCT is a more flexible property than treewidth in Chimera,
increasing flexibility to new generations of hardware, and also showed how to
use OCT to embed into a biclique virtual hardware. In the next section we
evaluate these new embedding subroutines against previously studied embedding
algorithms.

\section{Experimental Results}
\label{section:experimental_results}
In this section we experimentally evaluate virtual hardware against the
existing benchmark algorithms. First, we compare the approximation
\algorithmfont{Fast-OCT-Embed} against the exact \algorithmfont{OCT-Embed},
using no reduction routines. We then compare
the \algorithmfont{Reduced Fast-OCT-Embed} against Cai et al.'s Dijkstra-based
heuristic (denoted here as \algorithmfont{CMR (Dijkstra)}). Finally we conclude
with a comparison again Choi's \algorithmfont{TRIAD} algorithm for embedding
complete graphs. Against both benchmarks we find that \algorithmfont{Reduced
Fast-OCT-Embed} finds embeddings for larger graphs, using less qubits, with
fast run times (less than a second).

To minimize bias in the cross-algorithm comparisons, all algorithms and
subroutines (e.g. breadth-first search, Dijkstra's algorithm, etc.) were
implemented manually in C++ and are available at
\url{https://github.com/TheoryInPractice/aqc-virtual-embedding}.

\algorithmfont{OCT-Embed} is implemented using Lokshtanov et al.'s
simplification of Reed et al.'s iterative compression algorithm
\cite{lokshtanov2009simpler, reed2004finding}. \algorithmfont{Fast-OCT-Embed}
is computed using the smallest OCT number found with $10,000$ runs of
\algorithmfont{GreedyBipartite}; run times reported include the total run time
to collect this distribution. \algorithmfont{Reduced Fast-OCT-Embed}
additionally applies \algorithmfont{Qubit-Reduce} and
\algorithmfont{2Ex-Reduce} using
\algorithmfont{Fast-Qubit-Scoring}.

We implemented the \algorithmfont{CMR (Dijkstra)} algorithm from the
Dijkstra-based pseudocode provided on page 7 of \cite{cai2014practical}. Since
this heuristic does not necessarily produce an embedding if it exists, we run
the heuristic repeatedly until an embedding is found or the time cutoff is
reached; this provides the expected time to find an embedding.
\algorithmfont{TRIAD} is implemented with Choi's deterministic algorithm, and
\algorithmfont{Reduced TRIAD} uses the biclique virtual hardware with
\algorithmfont{Qubit-Reduce} and \algorithmfont{2Ex-Reduce} using
\algorithmfont{Fast-Qubit-Scoring}.

To provide a broad spectrum of comparisons, we generated problem graphs using
four random graph generators at three density levels (Table
\ref{table:densities}). While previous algorithms such as
\algorithmfont{CMR} have been tested on problem graphs with constant vertex
degree (e.g. grid and 3-regular graphs), this assumption is unrealistic for
real-world QUBOs. Intuitively, the complexity of the problem should scale with
the number of variables included. As an example, we note that Beasley's QUBOs
\cite{beasley1990or} have average vertex degree of approximately $\frac{n}{20}$
for $n$ vertex problems.

We define the random graph models as follows. Noisy bipartite graphs were
generated by splitting the vertices evenly (up to parity) into two partite
sets, including a bipartite
edge at probability $p$, and including a non-bipartite edge at probability
$\frac{p}{5}$. The GNP graphs (also known as Erd\H{o}s-R\'{e}nyi
\cite{erdos1960evolution}) are generated by flipping a coin for each possible
edge and including it with probability $p$. The regular graph generator samples
from the space of graphs where each vertex has degree exactly $k$.
Barab\'{a}si-Albert graphs \cite{albert2002statistical} are generated by
iteratively attaching $n - k$ vertices to a subgraph of $k$ vertices using
preferential attachment; we generate the initial subgraph using GNP with
$p=0.25$. All graphs are generated using the NetworkX implementations
\cite{hagberg2008exploring}, excepting Barab\'{a}si-Albert, which required a
modification to generate the initial subgraph as specified above.

\begin{table}[h!]
\centering
\begin{tabular}{rlll}
\toprule
\textbf{Graph Family} & \textbf{Low Density} & \textbf{Medium Density} & \textbf{High Density}\\
\midrule
Noisy Bipartite & $p = 0.25$ & $p = 0.50$ & $p = 0.75$\\
GNP & $p = 0.25$ & $p = 0.50$ & $p = 0.75$\\
Regular & $k = 0.25n$ & $k = 0.50n$ & $k = 0.75n$\\
Barabasi-Albert & $k = 0.25n$ & $k = 0.50n$ & $k = 0.75n$\\
\bottomrule
\end{tabular}
\caption{Definition of density levels for the random input graph generators.}
\label{table:densities}
\end{table}

All experiments were run on a workstation running Fedora 24, and were each
allocated a core on an Intel X5675 processor and 1GB of RAM.
Run times were limited to 60 minutes using the \texttt{timeout -k 10s 60m}
command, and no algorithm used more than its allocated memory.
The C++ code was compiled with g++ 5.3.1 at the \texttt{-O2} optimization
level, and controlled with wrapper scripts run with Python 2.7.11.
All experiments were seeded using the number of seeds specified in each
experiment below.
The data points plotted are the median over all problem graph instances and
seeded algorithm runs.

\subsection{Experimental Results}
\label{section:comparing_oct}

\begin{figure}[ht]
\centering
\includegraphics[width=\textwidth]{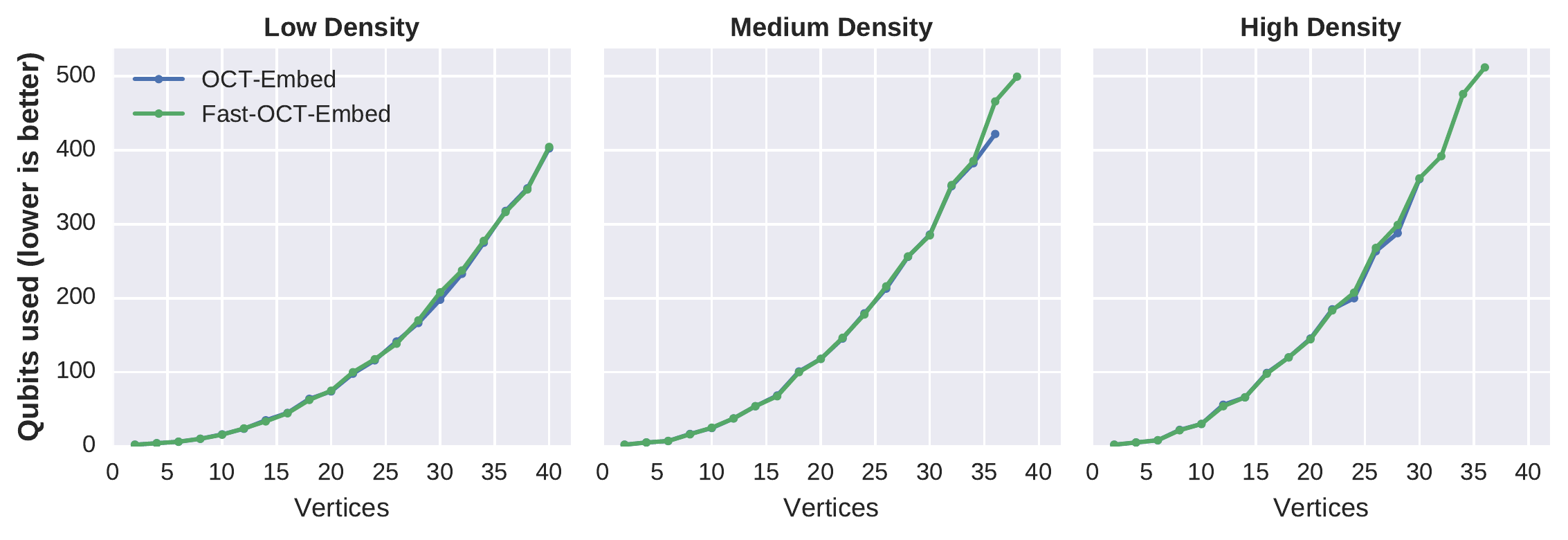}
\includegraphics[width=\textwidth]{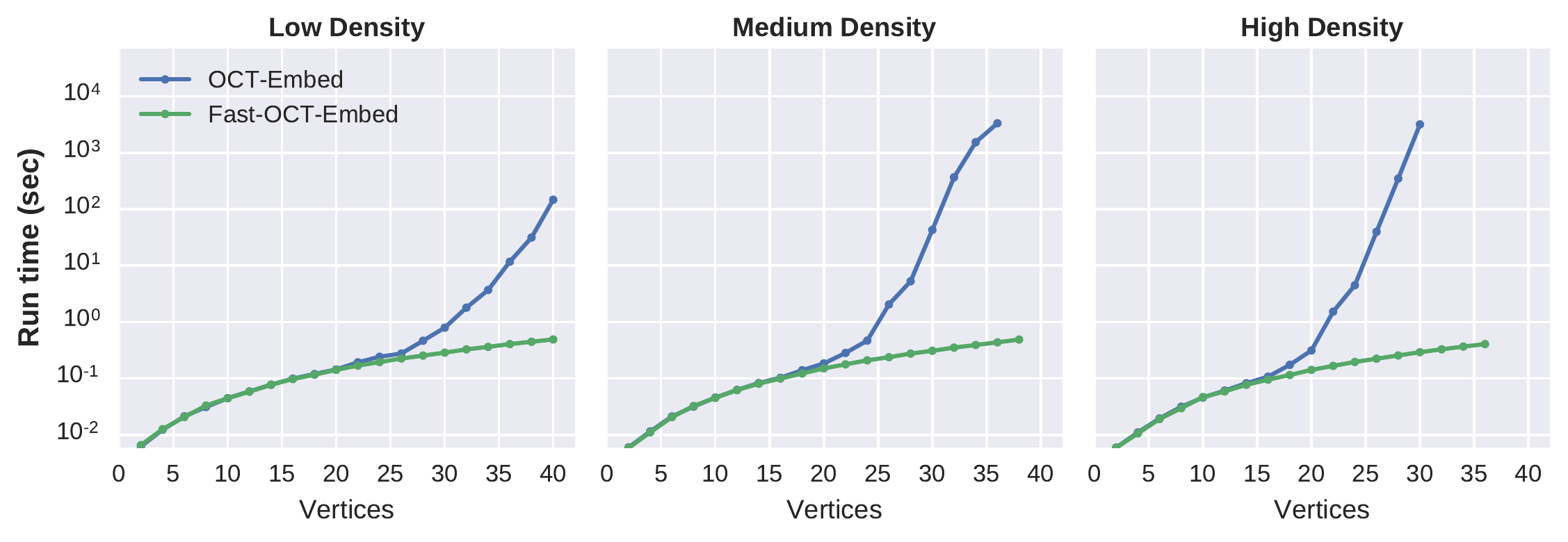}
\caption{Embedding GNP graphs into Chimera(4, 8, 8); data points are the median over 25 random graphs and 10 random algorithm seeds. Experimentally, we observe that the approximation algorithm performs notably better than its approximation factor guarantees, while additionally achieving highly practical run times.}
\label{figure:qubits_used_hybrid}
\end{figure}

Comparing \algorithmfont{OCT-Embed} and \algorithmfont{Fast-OCT-Embed} on 25
graph instances per $n$ value and 10 seeded algorithm runs, we find that
\algorithmfont{Fast-OCT-Embed} practically matches the solution quality of the
exact algorithm, while running in under a second. We report a representative
sample in Fig. \ref{figure:qubits_used_hybrid}.

\begin{figure}[h]
\centering
\includegraphics[width=\textwidth]{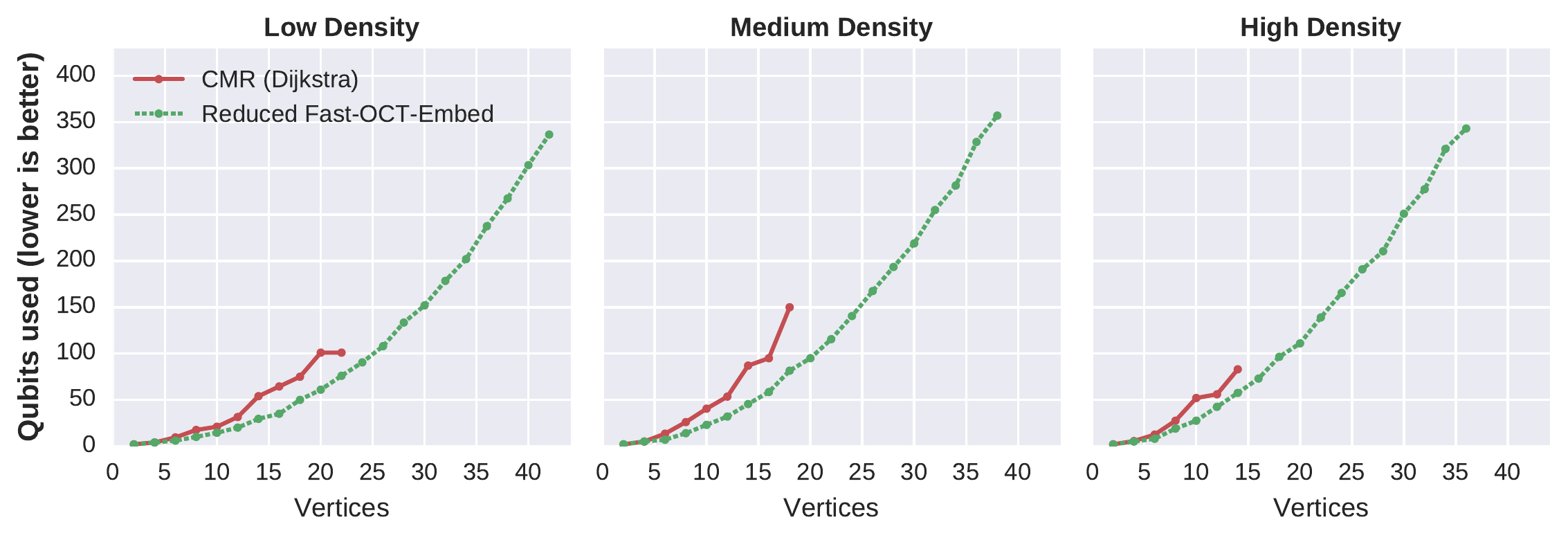}
\includegraphics[width=\textwidth]{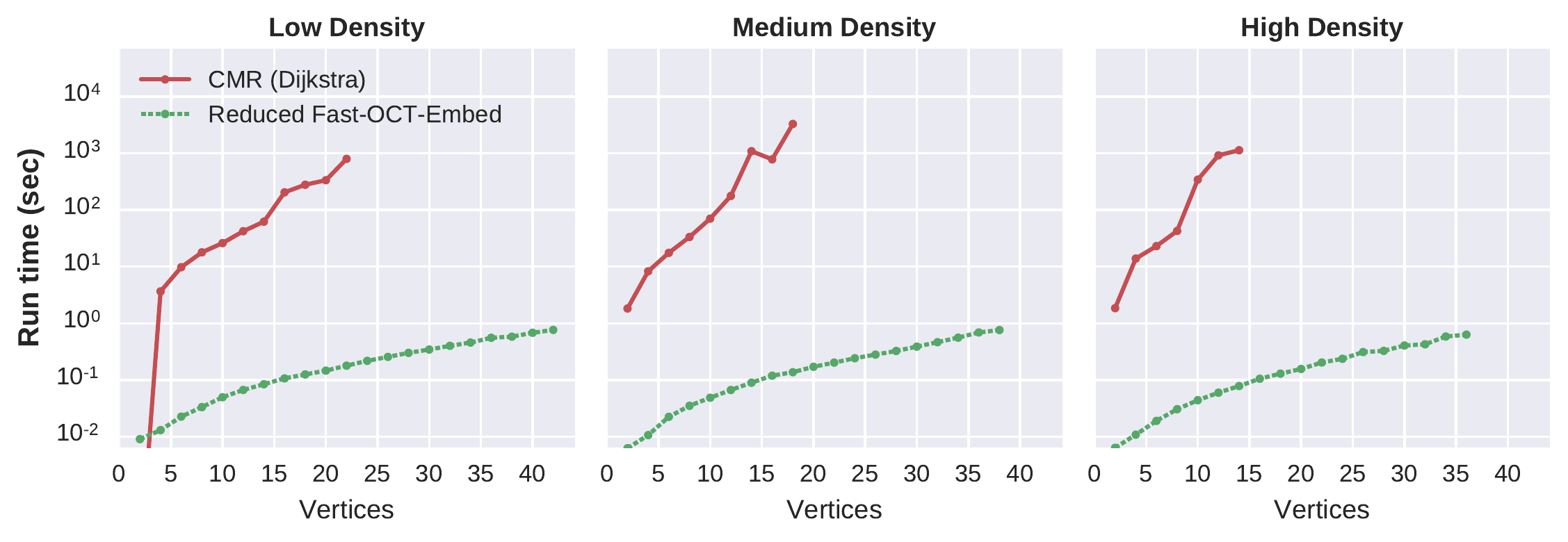}
\caption{Embedding GNP graphs into Chimera(4, 8, 8); data points are the median over 10 random graphs and 10 random algorithm seeds. \algorithmfont{Reduced Fast-OCT-Embed} consistently outperforms \algorithmfont{CMR} in both qubits used and run time.}
\label{figure:qubits_used_cmr}
\end{figure}

To maintain a reasonable run time while maintaining 10 graph instances per $n$, we reduced the comparison with \algorithmfont{CMR} to 10 seeded algorithm runs; this reduction did not impact the results since both \algorithmfont{CMR (Dijkstra)} and \algorithmfont{Reduced Fast-OCT-Embed}
restart automatically as needed. We found that \algorithmfont{CMR (Dijkstra)} could not find smaller embeddings than \algorithmfont{Fast-OCT-Embed}, in addition to having significantly longer run times. While \algorithmfont{CMR} may be competitive on very sparse graphs (e.g. grid graphs), we found that it was not competitive when the problem graph had a linear density. Fig. \ref{figure:qubits_used_cmr} contains a representative sample using GNP.

\begin{figure}[!h]
\centering
\includegraphics[width=\textwidth]{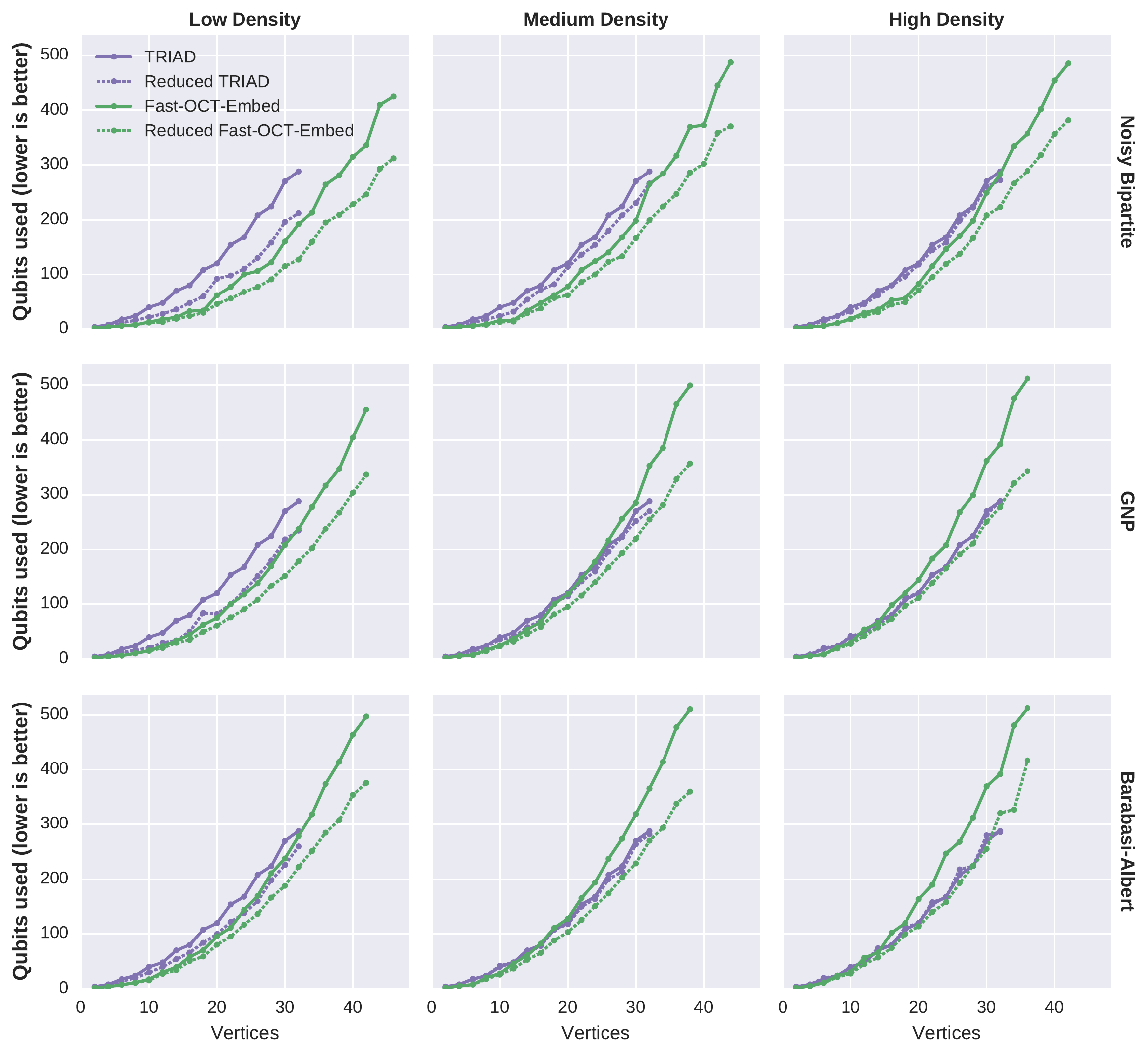}
\caption{Qubits used when embedding into Chimera(4, 8, 8); data points are the median over 25 random graphs and 10 random algorithm seeds. OCT-based algorithms consistently embed larger problem than possible with \algorithmfont{TRIAD}.}
\label{figure:qubits_used_triad}
\end{figure}

Our final comparison is against Choi's \algorithmfont{TRIAD} embedding
algorithm, the state-of-the-art for embedding highly dense problem graphs in
hardware without hard faults. We do not report run times, given that Choi's
algorithm is a
deterministic assignment and the OCT-based algorithm's run times are reported
in previous plots. For this experiment we again used 25 problem seeds and 10
algorithm seeds. Fig. \ref{figure:qubits_used_triad} contains a representative
sample. Again we find that \algorithmfont{Reduced Fast-OCT-Embed} embeds larger
graphs while using less qubits. We also note that \algorithmfont{Reduced TRIAD}
was effective compared to stock \algorithmfont{TRIAD} for all low density
graphs and some medium density graphs, while only adding less than a second to
the run time. Moreover, in several scenarios \algorithmfont{Reduced TRIAD}
performed better than vanilla \algorithmfont{OCT-Embed}, given the ``L''- vs.
``+''-shaped embeddings. However, the flexibility provided with ``+''-shaped
embeddings made the reduction subroutines much more effective, ultimately
producing a better full algorithm. As a best practice, then, we recommend that
embedding algorithm designers apply these standard reduction subroutines before
evaluating an embedding algorithm's effectiveness.

\section{Conclusion}
We have developed a virtual-hardware--based framework for constructing and
deploying optimized techniques for distinct parts of the minor embedding
process. By introducing a \emph{biclique virtual hardware}, we provide a
cleaner interface for embedding into the Chimera hardware layout and enable
modular subroutines for qubit reduction. Exploiting the bipartite structure in
problem graphs with odd cycle transversals, we are able to embed problems from
from a diverse set of generators and densities. Combining these two methods
leads to an embedding algorithm \algorithmfont{Reduced Fast-OCT-Embed} that
embeds larger problems, while using less qubits, for reasonably dense problem
graphs. Moreover, without any parallelization or system-specific tuning,
\algorithmfont{Reduced Fast-OCT-Embed} terminates in the order of seconds. This
algorithm sets a baseline for embedding dense problem graphs that should be
extended and tuned for the user's application.

Future extensions of this work could include tuned implementation of the
reduction methods, which are particularly promising for GPU parallelization.
Additionally, as the problem graph becomes highly dense,
we see that \algorithmfont{OCT-Embed} (by definition) converges to
\algorithmfont{TRIAD}. A more intricate embedding algorithm might not assume
the OCT vertices were a clique, allowing even more flexible embeddings.
Finally, adapting more intricate embedding algorithms (such as
\algorithmfont{CMR}) could provide even better improvements, but would require
significant development in the choice of relevant virtual hardware(s).

\section{Acknowledgements}
The authors would like to thank Steve Reinhardt and the anonymous two reviewers for feedback. This work is supported in part by the Gordon \& Betty Moore Foundation's Data-Driven Discovery Initiative through
Grant GBMF4560 to Blair D.~Sullivan, a National Defense Science \& Engineering Graduate Fellowship and a fellowship by the National Space Grant College and Fellowship Program and the NC Space Grant Consortium to Timothy D.~Goodrich.

\newpage
\bibliographystyle{spmpsci}
\bibliography{bib}

\newpage
\appendix
\section{Computing OCT in Series-Parallel Graphs}
\label{appendix:eppstein}

In this appendix we prove the following result:

\begin{proposition}
$OCT(G)$ can be computed in linear time for a series-parallel graph $G$.
\end{proposition}

The proof is based on the equivalence between series-parallel graphs and graphs
with nested ear decompositions. Using this decomposition, we show that a greedy
algorithm constructs a minimum OCT set. We start by defining series-parallel
graphs and nested ear decompositions:

\begin{definition}[Eppstein \cite{eppstein1992parallel}]
A graph $G$ is two-terminal series-parallel with terminals $s$ and $t$ if it
can be produced by a sequence of the following operations:
\begin{enumerate}
\item Base case: Create new graph, consisting of a single edge directed from
$s$ to $t$.
\item Parallel composition: Given two-terminal series-parallel graphs $X$ and
$Y$ with terminals $s_X$, $t_X$, $s_Y$, and $t_Y$, form a new graph
$G = P(X, Y)$ by identifying $s = s_X = s_Y$ and $t = t_X = t_Y$.
\item Series composition: Given two-terminal series-parallel graphs $X$ and
$Y$, with terminals $s_X$, $t_X$, $s_Y$, form a new graph $G = S(X, Y)$ by
identifying $s = s_X$, $t_X = s_Y$, and $t = t_Y$.
\end{enumerate}
\end{definition}

\begin{definition}[Ear Decomposition (Eppstein \cite{eppstein1992parallel})]
An \emph{ear decomposition} of an undirected graph $G$ is defined to be a
partition of the edges of $G$ into a sequence of \emph{ears}
$(E_1, E_2, \dots, E_k)$. Each ear is a path in the graph with the following
properties:
\begin{enumerate}
\item If two vertices in the path are the same, they must be two endpoints of
the path.
\item The two endpoints of each ear $E_i$, $i > 1$, appear in previous ears
$E_j$ and $E_j'$, with $j < i$ and $j' < i$.
\item No interior point of $E_i$ is in $E_j$ for any $j < i$.
\end{enumerate}
\end{definition}

\begin{definition}[Nest Intervals (Eppstein \cite{eppstein1992parallel})]
Given an ear decomposition $(E_1, E_2, \dots, E_k)$, we say that $E_i$ is
\emph{nested} in $E_j$ if both endpoints of $E_i$ are contained in $E_j$. The
\emph{nest interval} of $E_i$ in $E_j$ is the path in $E_j$ between the two
endpoints of $E_i$.
\end{definition}

\begin{definition}[Nested Ear Decomposition (Eppstein
  \cite{eppstein1992parallel})]
An ear decomposition is \emph{nested} if the following hold:
\begin{enumerate}
\item For each $i > 1$ there is some $j < i$ such that $E_i$ is nested in $E_j$.
\item If two ears $E_i$ and $E_{i'}$ are both nested in the same ear $E_j$,
then either the nest interval of $E_i$ contains that of $E_{i'}$ or vice versa.
\end{enumerate}
\end{definition}

\noindent Eppstein's result shows that these two graph classes are equivalent:

\begin{theorem}[Eppstein \cite{eppstein1992parallel}]
Any undirected two-terminal series-parallel graph has a nested ear
decomposition starting with a path between the terminals, and any undirected
graph with a nested ear decomposition is two-terminal series-parallel with its
terminals being the endpoints of the first ear.
\end{theorem}

Furthermore, Eppstein shows that these decompositions can be computed in
$O(\log^2(n))$ time on a parallel computer, therefore computing the
decomposition itself will not bottleneck an OCT-computing algorithm. We now
show that given a nested ear decomposition, we can greedily compute a minimum
OCT set. First, we define the parity of ears.

\begin{definition}[Ear Parity]
We say that an ear $E_i$ is \emph{odd} if the number of vertices in $E_i$ and
its nest interval sum to an odd number. We define an \emph{even ear}
analogously.
\end{definition}

Next, we want to show that we can compute the minimum OCT set on a single nest
interval correctly.

\begin{lemma}
Given an ear decomposition $(E_1, \dots, E_k)$, let $E = (E_i, \dots, E_j)$ be
an ordered, maximal list of ears contained in a single nest interval. Then the
minimum number of OCT vertices contained in these ears is number of maximal
in-order sublists of $E$ composed only of odd ears.
\label{lemma:maximal_nest}
\end{lemma}

\begin{proof}
We proceed by induction on the number of sublists. Suppose that there are zero
maximal sublists of odd ears, therefore every ear is even. Then every path from
the left-most vertex on the nest interval to the right-most vertex on the nest
interval will have the same parity, and we are able to two-color these cycles
and the minimum OCT number is zero. Suppose instead that there is one maximal
sublist of odd ears, therefore all ears are odd. Removing an endpoint from the
inner-most odd cycle renders the remaining edges of this smallest nest interval
into bridges that cannot be part of a cycle. This removal also breaks all odd
cycles in the maximal nest interval, because any cycle on the remaining ears
must use the vertices from two odd cycles of length $2x+1$ and $2y+1$, minus
the length of the smallest nest interval twice, leaving an odd number of
vertices in the cycle. Again we can two-color these and we are done.

Suppose we have an interval with $k$ maximal sublists. If there are any even
ears on the outside then we can remove them using the first base case. We now
find the inner-most odd ear that is outside of every remaining even ear.
Removing an endpoint from this ear renders the outer odd ears bipartite by the
second base case. Applying the inductive hypothesis to the earlier ears finds
$k-1$ OCT vertices, therefore we have found a total of $k$ OCT vertices.
\smartqed \qed
\end{proof}

\begin{corollary}
We can compute the minimum OCT number of the ears contained in a single maximal
nest interval in linear time.
\end{corollary}
\begin{proof}
In the above proof we visited each ear once.
\smartqed \qed
\end{proof}

\begin{proposition}
$OCT(G)$ can be computed in linear time for a series-parallel graph $G$.
\end{proposition}

\begin{proof}
We proceed by induction on the number of maximal nest intervals. If there are
no nest intervals then we have a single ear and are done, the graph is (by
definition) bipartite. Otherwise there is some nest interval. Applying Lemma
\ref{lemma:maximal_nest}, we can compute a minimum OCT set. Since every nest
interval is disjoint, by definition, we can apply the inductive hypothesis to
compute the minimum OCT set of the other intervals, visiting each interval
exactly once. The number of intervals is bounded by the number of vertices,
therefore we compute a minimum OCT set in linear time.
\smartqed \qed
\end{proof}

\end{document}